\documentclass[12pt]{article}

\usepackage{etex} 

\usepackage[margin=3cm]{geometry}
\usepackage{amsmath}
\usepackage{amssymb}
\usepackage{amsthm}
\usepackage{fancyhdr}
\usepackage{verbatim}
\usepackage{mathtools}
\usepackage{color}
\usepackage[utf8]{inputenc}

\numberwithin{equation}{section}

\theoremstyle{plain}
\newtheorem{theorem}{Theorem}[section]
\newtheorem{definition}[theorem]{Definition}
\newtheorem{lemma}[theorem]{Lemma}

\newtheorem{proposition}[theorem]{Proposition}

\theoremstyle{definition}

\newcommand{\Z}{\mathbb Z}
\newcommand{\N}{\mathbb N}

\begin{document} 

\title{On non-stability of one-dimensional non-periodic ground states} 
\author{Damian G\l odkowski \\  Institute of Mathematics \\ Polish Academy of Sciences \\ \'Sniadeckich 8,  00-656 Warsaw, Poland\\
d.glodkowski@uw.edu.pl
\\ \\ Jacek Mi\c{e}kisz \\ Institute of Applied Mathematics and Mechanics \\ University of Warsaw \\ Banacha 2, 02-097 Warsaw, Poland \\ miekisz@mimuw.edu.pl} 

\baselineskip=20pt
\maketitle 

\begin{abstract}
We address the problem of stability of one-dimensional non-periodic ground-state configurations with respect to finite-range perturbations of interactions 
in classical lattice-gas models. We show that a relevant property of non-periodic ground-state configurations in this context is their homogeneity. 
The so-called strict boundary condition says that the number of finite patterns of a configuration have bounded fluctuations on any finite subsets of the lattice.
We show that if the strict boundary condition is not satisfied, then in order for non-periodic ground-state configurations to be stable, 
interactions between particles should not decay faster than $1/r^{\alpha}$ with $\alpha>2$. In the Thue-Morse ground state, 
number of finite patterns may fluctuate as much as the logarithm of the lenght of a lattice subset. We show that the Thue-Morse ground state is unstable 
for any $\alpha >1$ with respect to arbitrarily small two-body interactions favoring the presence of molecules consisting of two spins up or down.
We also investigate Sturmian systems defined by irrational rotations on the circle. They satisfy the strict boundary condition but nevertheless
they are unstable for $\alpha>3$.
\end{abstract}

\section{Introduction}

Since the discovery of quasicrystals \cite{shechtman}, one of the problems in statistical mechanics is to construct microscopic models 
of interacting atoms or molecules in which all configurations minimizing energy, the so-called ground-state configurations, are non-periodic. 
Here we will discuss models in which although all ground-state configurations are non-periodic, they all look the same, they cannot be distinguished locally. 
More precisely, they support the unique translation-invariant measure called the non-periodic ground state of the system. 

There were constructed many classical lattice-gas models without periodic ground-state configurations \cite{Rad0,cmpmiekisz,Miesta,strictboundary}.
Some of these models are based on non-periodic tilings of the plane with Wang square-like tiles \cite{robinson}. 
In such tilings, centers of tiles form a regular two-dimensional lattice $\Z^{2}$,
assignments of tiles to vertices are non-periodic. Now, types of tiles are identified with types of particles, interactions between particles correspond to matching rules.
Namely, the interaction energy between two particles which match as tiles is $0$, if they don't match, the energy is positive, say $1$. 
It is easy to see that ground-state configurations correspond to tilings, therefore there are no periodic ones.  
Interactions in such models are obviously non-frustrated - all interactions attain their minima (equal to $0$) in ground-state configurations.  

A desired property of non-periodic ground-state configurations is their stability against small perturbations of interactions between particles. 
For two-dimensional systems with finite-range non-frustrated interactions and with unique non-periodic ground-state measure, 
the relevant property is the so-called strict boundary condition -the requirement that the number of finite patterns 
can fluctuate at most proportional to the boundary of the lattice subset (a precise definition is given in Section 2) \cite{strictboundary}.
It was shown in \cite{strictboundary} that the strict boundary condition is equivalent to the stability of ground-state configurations. 
More precisely, non-periodic ground states are stable against small perturbations of the range $d$ if and only if the strict boundary condition is satisfied for all local patterns 
of sizes smaller than $d$. So far we do not have examples of two-dimensional finite-range classical lattice-gas models without periodic ground-state configurations 
which satisfy the strict boundary condition. For example, it was proved that the non-periodic ground state based on Robinson's tilings is unstable 
with respect to arbitrarily small chemical potentials \cite{jaradin}.  

One of the goals of our paper is to show that the strict boundary condition is relevant also for infinite-range interactions.

Situation concerning non-periodicity is quite different in one-dimensional models. 
It is known that one-dimensional systems without periodic ground-state configurations require infinite-range interactions \cite{bundangnenciu,schulradin,thirdlaw}. 
It follows that one-dimensional non-periodic ground states cannot be stable with respect to small perturbations in any reasonable space of infinite-range interactions,
in the space $l_{1}$ of summable interactions we can cut the tail of an arbitrary small $l_{1}$-norm to obtain a finite-range Hamiltonian with some periodic
ground-state configurations.

One-dimensional two-body interactions producing only non-periodic ground-state configurations were presented in \cite{bakbruinsma,aubry2,aubry3,jedmiek1,jedmiek2}. 
Hamiltonians in these papers consisted of two-body repelling interactions between particles and a chemical potential favoring particles.
Such interactions are obviously frustrated. Ground states of these models form Cantor sets called devil's staircases. Non-periodicity is present only for certain values 
(of measure zero) of chemical potentials - an arbitrary small change of a chemical potential destroys the non-periodic ground state.

In \cite{tmhamiltonian}, a non-frustrated infinite-range, exponentially decaying four-body Hamiltonian was constructed, with the unique ground-state-measure supported 
by Thue-Morse sequences \cite{keane}. Recently in \cite{ahj} there were constructed non-frustrated two-body (augmented by some finite-range interactions) 
Hamiltonians producing exactly the same ground states as in the frustrated model of \cite{bakbruinsma,aubry2,aubry3}. 

Here we investigate the stability of non-periodic ground-state configurations with respect to finite-range perturbations.
We show that strict boundary condition plays here an important role. Our general result is that if the strict boundary condition is not satisfied 
in order for non-periodic ground-state configurations to be stable, interactions between particles should not decay faster than $1/r^{\alpha}$ with $\alpha>2$,
see Theorem 2.2. 

In the Thue-Morse ground state, the number of finite patterns may fluctuate as much as the logarithm of the lenght of a lattice subset.
We show that such a ground state is unstable with respect to arbitrarily small two-body interactions favoring 
the presence of molecules consisting of two spins up or down, see Theorem 4.1.

We also investigate Sturmian systems defined by irrational rotations on the circle, which satisfy the strict boundary condition \cite{ahj}. 
Hamiltonians having Sturm sequences as ground state-configurations we recently constructed in \cite{ahj}. We show that if $\alpha>3$, 
then Sturmian systems are not stable, see Theorem 5.6.

\section{Strict boundary condition and non-stability of non-periodic ground states}

A frequency of a finite pattern in an infinite configuration is defined as the limit of the number of occurrences of this pattern in a segment of length $L$ 
divided by $L$ as $L \rightarrow \infty$. All sequences in any given Sturmian system have the same frequency for each pattern. We are interested now whether the fluctuations 
of the numbers of occurrences are bounded (bounded by the boundary of the size of the boundary, which in one-dimensional systems is equal to $2$). 
If that is the case, configurations are said to satisfy the {\bf strict boundary condition} \cite{strictboundary} or rapid convergence of frequencies to their equilibrium values
\cite{peyriere,gambaudo}. 

\begin{definition}
Given a sequence $X=(x_n) \in  \{0,1\}^\Z$ and a finite word $w$, define the frequency of $w$ as 
\[
\xi_w=\lim_{N\to \infty}\frac{\#\{|n|\le N\mid x_n\dots x_{n+|w| - 1}=w\}}{2N}.
\]
Furthermore, for a segment $A\subset \Z$, denote by $X(A)$ the sub-word $(x_n)_{n\in A}$. 
We say that a sequence $X$ satisfies the {\bf strict boundary condition} (quick convergence of frequencies) if for any word $w$ and a segment $A \subset \Z$, 
the number of appearances of $w$ in $X(A)$, $n_{w}(X(A))$, satisfies the following inequality:
\[
|n_{w}(X(A)) - \xi_{w}|A|| < C_{w},
\] 
where $C_{w}>0$ is a constant which depends only on the word $w$.
\end{definition}

We consider classical lattice-gas models with unique ground-state measures supported by non-periodic ground-state configurations. 
The configuration space of a system is denoted by $\Omega = \{0, 1\}^{\Z}$, where $0$ means the absence of a particle at a given lattice site, 
and $1$ its presence. Let $f(r) = 1/r^{\alpha}$ be the energy of one-dimensional interaction between particles at a distance $r$ in a classical lattice-gas model. 

\begin{theorem}
If non-periodic ground-state configurations do not satisfy the strict boundary condition and the interaction energy decays as $1/r^{\alpha}$ with $\alpha>2$, 
then they are unstable with respect to an arbitrary small chemical potential - a one-body on-site interaction.
\end{theorem}

\begin{proof}
Let $X \in \Omega$ be a ground-state configuration. Let us assume that the system does not satisfy the strict boundary condition. 
It means that for any $C>0$, there are two segments of consecutive lattice sites, $S_{1}, S_{2} \subset \Z$ of the length $L$, such that  
$n(X(S_{2})) - n(X(S_{1})) > C$, where $n(X(S_{i}))$ is the number of 1's in $X$ on $S_{i}$, $i=1, 2$. 
We construct a finite excitation $Y$ of $X$, that is a configuration $Y$ which differs from $X$ only on a finite number of lattice sites. 
Namely, let $Y=X$ outside $S_{1}$ and $Y$ on $S_{1}$ is equal to $X$ on $S_{2}$.
Let us now introduce an on-site interaction - a chemical potential which favors the presence of particles - it assigns to each particle a negative energy $-\mu$
for some small $\mu >0$. We will be concerned with the relative hamiltonian $H(Y|X) = H(Y) - H(X)$. We will show that for $\alpha >2$, 
$H(Y|X)<0$ that is by a finite change of $X$ one can decrease the energy hence $X$ is not a ground-state configuration for a perturbed hamiltonian.

Obviously, the chemical potential decreases the energy by $C\mu$. Now we have to bound appropriately the possible increase 
of the energy associatet with a two-body original interactions. The increase of the energy can be divided into two parts: 
$E_{1}$ associated with interactions between particles at a distance smaller than $L$ and $E_{2}$ associated with interactions between particles 
at a distance equal or bigger than $L$. Now we have,

\begin{equation}
E_{1}\leq 2\sum_{r=1}^{L}\frac{r}{r^{\alpha}} < \int_{x=1}^{L}\frac{2}{x^{(\alpha-1)}}dx +2 = \frac{2}{(2-\alpha)}L^{2-\alpha} +2 -\frac{2}{2-\alpha},
\end{equation}

\begin{equation}
E_{2}\leq 2L\sum_{r=L}^{\infty}\frac{1}{r^{\alpha}} < 2L(\int_{L}^{\infty}\frac{1}{x^{\alpha}}dx + \frac{1}{L^{\alpha}}) = \frac{2}{\alpha-1}L^{2-\alpha}+2L^{1-\alpha}.
\end{equation}

It follows that for $\alpha>2$, $E_{1} + E_{2} < C\mu$ and therefore $H(Y|X)<0$ and therefore $X$ is not a ground-state configuration for the perturbed interaction.
\end{proof}

\section{Toeplitz period-doubling ground state}

We construct Toeplitz  \cite{toeplitz} (also known as period-doubling configurations) in the following way.
We place $-1$ (a symbol corresponding to the absence of a particle) on a sublattice $L_{1} \subset \Z$ of period $2$. 
Then we place $1$ (a symbol corresponding to a particle) on a sublattice $L_{2} \subset \Z$ of period $4$ which is disjoint from $L_{1}$.
We repeat this procedure ad infinitum and get a Toeplitz configuration $X \in \Omega =\{-1, 1\}^{\Z}$ such that $X(i)=(-1)^{j}$ if $i \in L_{j}$.
$X$ is obviously non-periodic, there are particles on every sublattice $L_{j}$  for even $j's$. 

It is easy to see that the closure of the orbit of $X$ by translations $T$ supports the unique translation-invariant measure, we denote it by $\rho_{To}$.
In this way we have constructed a uniquely ergodic dynamical system $(\Omega, \rho_{To}, T)$. 
The density of particles ($1$'s) in $X$ is equal to $1/3$. Now we will show that $X$ does not satisfy the strict boundary condition.  

Let us look at particles on sublattices $L_{j}, j \leq m$.
One can find a segment $W \subset \Z$ such that  $i \in L_{2}$ is the first site of $W$, $X(i)=1$, $i+2 \in L_{4}$ so  $X(i+2)=1$, $i+2+8 \in L_{6}$ 
so $X(i+2+8)=1$, ... , $i+ 2 + 8 + ...+ 2 \times 4^{(m/2)-1} \in L_{(m/2)-1}$ so $X(i+2 + 8 + ...+ 2 \times 4^{(m/2)-1})=1.$  
Let the length of $W$ be equal to
$2(2 + 8 + ...+ 2 \times 4^{(m/2)-1}) = \frac{4}{3}(4^{m/2}-1)$.

Hence the number of particles in $X$ on sublattices $L_{j}, j \leq m$ is equal to $n_{1}+ ...+n_{m/2}$,
where $n_{m/2}=2$ and $n_{k-1}= 4n_{k}-2, k=m/2, m/2-1, m/2-2,...,1$. 

Now we take $X(W)$ and place it on $V = T_{a}(W)$, where $T_{a}$ is a shift operator to the right by
$a = 4\sum_{k=0}^{(m/2) -2}4^{k} = \frac{4}{3} 4^{(m/2)-1}$.
We see that the number of particles in $X(V)$ decreases by $m/2$ with respect to $X(W)$, one particle for each $L_{j}, j \leq 2k, k =1, ... ,m/2$.
 
Obviously, the strict boundary condition is not satisfied. Moreover, we can apply the above procedure for any $m \leq V$
so the fluctuations of the number of particles on the $V \subset Z$ can be of the order $\log_{4} |V|$ where $|V|$ is the length of $V$. 

Let us assume that the Toeplitz measure $\rho_{To}$ is the unique ground state of some two-body interactions decaying as $1/r^{\alpha}$, 
where $r$ is the distance between particles.

\begin{theorem}
The Toeplitz ground state $\rho_{To}$ is unstable against an arbitrarily small chemical potential which favors the presence of particles. 
\end{theorem}
\begin{proof}

To prove the instability of the ground state $\rho_{To}$ we introduce a chemical potential $\mu$ favoring the presence of particles and a local perturbation $Y$ of $X$
such that $H(Y|X) <0$. We take $X(W)$, described above, and place it on $V$, a certain shift of $W$, such that sublattices $L_{j},  j \leq n$ (for some even $n$ to be chosen later)
agree in $Y'$ on $V$ and on the complement of $V$. In this way we introduced an extra $1$ on all sublattices $L_{j}, n< j < \log_{4} V$ with even $j's$.

Now we will construct an upper bound for $H(Y|X)$. Denote by $H_{4^{n}}$ the energy of interactions of the first particle on the left side of $L_{n+2}$  
with the particles on the part of the complement of $V$ left to the particle. We get 

\begin{equation}
H_{4^{n}}  <   \sum_{k=4^{n}}^{\infty}\frac{1}{k^{\alpha}} < 2\int_{4^{n}}^{\infty}\frac{1}{x^{\alpha}}dx = \frac{2}{(\alpha-1)4^{n(\alpha-1)}}.
\end{equation}

It follows that 

\begin{equation}
H(Y|X)  <  2 \sum_{k=1}^{V/4^{n}}H_{k4^{n}}  <  2\sum_{k=1}^{V/4^n} \frac{1}{(k4^{n})^{(\alpha-1)}} <  \frac{4V^{2-\alpha}}{(\alpha-1)(2-\alpha)4^{n}}
\end{equation}

for $1< \alpha <2$.

Let $n$ be a minimal even number such that

\begin{equation}
\frac{4V^{2 - \alpha}}{(\alpha -1)(2-\alpha)4^{n}} < 1
\end{equation}

so $n < (2-\alpha) \log_{4}|V| - \log_{4}(\alpha-1)(2-\alpha) + 3.$

It follows that the two-body interaction energy is bounded (independent on $|V|$) and the number of excessive $1's$ is bigger than  

\begin{equation}
\log_{4}|V|  -  (2-\alpha) \log_{4}|V| + \log_{4}(\alpha-1)(2-\alpha)  - 3 > \frac{\epsilon}{2} \log_{4}|V|
\end{equation}

for any $\alpha = 1 + \epsilon$  and a sufficiently big $\epsilon$ dependent $V$.

It shows that $H(Y|X)<0$ so $X$ is not a ground-state configuration for any arbitrarily small $\mu$.

\end{proof}

\section{Thue-Morse ground state}

We prove here that the Thue-Morse ground state is unstable with respect to arbitrarily small two-body interactions.

We begin by constructing a one-sided Thue-Morse sequence. We put $1$ at the origin and perform successively the substitution $S$: $ 1 \rightarrow 10, 0 \rightarrow 01$.
In this way we get a one-sided sequence $1001 0110 0110 1001 ...$, $\{X_ {TM }(i)\}, i \geq 0$. We define $X_{TM} \in \Omega = \{0, 1\}^{\Z}$ by setting 
$X_{TM}(i) = X_{TM} (-i-1 )$ for $i<0$. Let $T$ be the translation operator, i.e., $T:\Omega \rightarrow \Omega, (T(X))(i) = X(i-1), X \in \Omega$. 
Let $G_{TM}$ be the closure (in the product topology of the discrete topology on ${0,1}$) of the orbit of $X_{TM}$ by translations, 
i.e., $G_{TM} = \{T^{n} (X_{TM}), n \geq 0\}^{cl}.$ It can be shown \cite{keane} that $G_{TM}$ supports exactly one translation-invariant probability measure 
$\mu_{TM}$ on $\Omega$. 

Let us identify now $1$ with $+1$ and $0$ with $-1$, so particles are represented by spins up and the empty spaces by spins down.
It was shown in \cite{tmhamiltonian} that $\rho_{TM}$ is the only ground state of the following exponentially decaying four-spin interactions,

\begin{equation}
H_{TM} = \sum_{r=0}^{\infty} \sum_{p=0}^{\infty} H_{r,p},
\end{equation}

where

\begin{equation}
H_{r,p} = \sum_{i \in \Z} J(r,p) (\sigma_{i} + \sigma_{i+2^{r}})^{2} (\sigma_{i+(2p+1)2^{r}} + \sigma_{i+(2p+2)2^{r}})^{2}
\end{equation}

and $\sigma_{i}(X)= X(i) \in \{+1, -1\}$.

\begin{theorem}
The Thue-Morse ground state $\rho_{TM}$ is unstable against an arbitrarily small chemical potential 
which favors the presence of molecules consisting of two spins up or down.
\end{theorem}
\begin{proof}

Let $X \in \{+1,-1\}^{\Z}$ be any Thue-Morse sequence and let $Y(i) = X(i) X(i + 1)$. It is easy to see that $Y$ is a Toeplitz sequence. 
Let us now consider $11$ and $00$ as $A$ and $10$ $01$ as $B$ type molecules, respectively. $A$ and $B$ molecules form a Toeplitz sequence
and therefore their numbers on the segment $V$ may fluctuate on the order of $\log_{4}|V|$. Now we introduce a chemical potential $h$ favoring $A$-type molecules.
It follows (in the same way as in the Toeplitz case) that if we take into account interactions $J(0,p)$ and such that $J(0,p)$ decays as $1/p^{\alpha}$, 
then the Toeplitz ground state is unstable with respect to any arbitrarily small $\mu$ for any $\alpha >1$. 

Now we have to take care of $J(r,p)$ for all $r\geq 1$. Again we assume that $J(r,p)$ decays as $1/[(2p+2)2^{r}]^{\alpha}$.
We will use the fact that Thue-Morse sequences are self-similar. 
Namely, we can group two successive symbols $10$ and $01$ and replace them by $1$ and $0$ respectively and in this way we get again a Thue-Morse sequence. 
One can do analogous groupings on every scale, for example $1001 \rightarrow 1$, $0110 \rightarrow 0$. We will also use the structure of our interaction.
First we notice that for any sequence of successive blocks of $10$ and $01$, no two pairs of either $11$ or $00$ are at an odd distance, 
therefore interactions in the Hamiltonian for $r=0$ and any $p$ attains the zero value. It follows from the self-similarity of Thue-Morse sequences
that for any sequence of successive blocks of $1001$ and $0110$, no two pairs of either two $1$'s or two $0$'s at a distance $2$ are at a distance $2(2p+1)$
for any $p>0$ so the Hamiltonian for $r=1$ and any $p$ attains the zero value. In general, for any sequence of successive blocks of $S^{r}(1)$ and $S^{r}(0)$,
the hamiltonian associated with any $r \geq 0$ and any $p$ attains the zero value, $H_{r,p}(Y|X)=0$.

Now, to prove the instability of the Thue-Morse ground state we mimic the procedure used in the Toeplitz case.
Namely, let $X$ be a Thue-Morse sequence. We construct $Y$, a local perturbation of $X$, in the following way. 
Let $V \subset \Z$ be a segment of $\Z$, then $Y=X$ on the complement of $V$ and we put on $V$ an appropriate block of $X$ such that 
$Y$ consists of successive blocks of $S^{r}(1)$ and $S^{r}(0)$ for any $r<r^*$ (to be chosen later) so $H_{r,p}(Y|X) = 0$ 
for any $p$ and $r<r^*$ and the number of excessive $A$-molecules is bigger than $\log_{4} |V| - r^*$.

Now fix $r \geq r^*$ and consider interactions $J(r,p), p>0$. Similar calculations as in the Toeplitz case show that for

\begin{equation}
H_{r} = \sum_{p=0}^{\infty}H_{r,p}
\end{equation}

we have

\begin{equation}
H_{r}(Y|X) < \frac{V^{2-\alpha}}{(\alpha-1)(2-\alpha)2^{r}}
\end{equation}

for $1 < \alpha < 2$ and therefore 

\begin{equation}
H(Y|X) = \sum_{r \geq r^{*}} H_{r}(Y|X) < \frac{V^{2-\alpha}}{(\alpha-1)(2-\alpha)2^{r^{*}-1}}.
\end{equation}

We choose a minimal $r^{*}$ such that $\frac{V^{2-\alpha}}{(\alpha-1)(2-\alpha)2^{r^{*}-1}} < 1$.

Hence we get that $H(Y|X) < 0$ which shows the instability of the Thue-Morse ground state.
\end{proof}

\section{Sturmian ground states}

We will consider bi-infinite sequences (words) of two symbols $\{0,1\}$, i.e. elements of $\Omega = \{0,1\}^\Z$.
We will identify the circle $C$ with $R/\Z$ and consider an irrational rotation by $\varphi$ (which is given by translation on $R/\Z$ by $\varphi \mod 1$).   

\begin{definition}
Given an irrational $\varphi\in C$ we say that $X\in\{0,1\}^\Z$ is \textbf{generated by $\varphi$} if it is of the following form: 
\begin{equation*}
    X(n) = \begin{cases}
               0               & \text{when} \ x+n\varphi \in P \\
               1                & \text{otherwise}
           \end{cases}
\end{equation*}
where $x\in C$ and $P=[0,\varphi)$.
\end{definition}

We call such $X$ a Sturmian sequence corresponding to $\varphi$. Let $T$ be the translation operator.
Let $G_{St}$ be the closure of the orbit of $X$ by translations. It can be shown that $G_{St}$ supports exactly one translation-invariant probability measure 
$\rho_{St}$ on $\Omega$. 

From now on we will consider only rotations by badly approximable numbers

\begin{definition}\label{badly-def}
We say that a number $\varphi$ is \textbf{badly approximable} if there exists $c>0$ such that 
$$\left | \varphi - \frac{p}{q} \right | > \frac{c}{q^2} $$
for all rationals $\frac{p}{q}$. 
\end{definition}

We need one more characterization of Sturmian words in terms of patterns that don't appear in given sequence.

\begin{theorem}\label{forbidden}\emph{\cite[Theorem 4.1]{ahj}}
Let $\varphi\in (\frac{1}{2},1)$ be irrational. Then there exist a natural number $m$ and a set $F\subseteq \N$ of forbidden distances such that Sturmian words generated by $\varphi$ are uniquely determined by the absence of the following patterns: $m$ consecutive 0's and two 1's separated by distance from $F$.
\end{theorem}

To characterize Sturmian words generated by irrationals from $(0,\frac{1}{2})$ we have to change the roles of 0's and 1's. We will show that $F$ can also be described by the rotation. 

\begin{proposition}\label{forbidden-pro}
The set $F$ from Theorem \ref{forbidden} may be chosen in the following way: 
$$k\notin F \iff \exists y\in [\varphi , 1) \ y+k\varphi \in [\varphi, 1),$$
or equivalently 
$$ k\in F \iff k\varphi \in [1-\varphi, \varphi].  $$
\end{proposition}
\begin{proof}
We start with the proof of equivalence of the above statements. Let's see that 
$$\neg (\exists x\in [\varphi , 1) \ y+z \in [\varphi, 1)) \iff \forall x\in[\varphi, 1) \ y+z\in [0,\varphi) \iff \{y+z: y\in [\varphi,1)\} \subseteq [0,\varphi). $$
The arc $\{y+z: y\in [\varphi,1)\}$ is contained in $[0, \varphi)$ if and only if its endpoints are in $[0,\varphi]$ which means that $z \in [1-\varphi, \varphi]$. Picking $z=k\varphi$ completes the proof of this part. 

For the proof that $F$ is a good set of forbidden distances proceed as follows. Fix a Sturmian word $X$ generated by $\varphi$ with an initial point $x$. If $X(n)=X(n+k)=1$ then $x+n\varphi\in [\varphi,1)$ and $ x+(n+k)\varphi \in [\varphi,1)$, so for $y=x+n\varphi$ we have $y, y+k\varphi \in [\varphi,1)$ which proves that $k\notin F$. Conversely, assume that $k\notin F$. Let $y$ be such that $y, y+k\varphi \in [\varphi,1)$. Then $Y(0)=Y(k)$ where $Y$ is the sequence generated by $\varphi$ with initial point $y$, so $k$ is not a forbidden distance. 
\end{proof}

Given set of forbidden distances we may easily construct non-frustrated Hamiltonians for which the unique ground-state consists exactly of Sturmian words generated by $\varphi$. Simply we need to assign positive energies to all forbidden patterns and zero otherwise (for more details see \cite[Theorem 5.2]{ahj}). 

We need a general lemma concerning Sturmian sequences. 
\begin{lemma}\label{lemma-doubling}
Let $\varphi \in (0,1)$ be irrational. Then for each $m \in \N$ there is $k\geq m$ and a finite word $w\in\{0,1\}^{\{0,1,\dots, k-1\}}$ of length $k$ such that 
\begin{itemize}
    \item $w$ is a subword of a Sturmian word generated $\varphi$, 
    \item doubling of $w$, $\widetilde{w}\in \{0,1\}^{\{0,1,\dots, 2k-1\} }$ given by
    $$\widetilde{w}(i)= \widetilde{w}(i+k) =w(i) \ \text{for} \ i< k $$
    also is a subword of a Sturmian word generated by $\varphi$.
\end{itemize}
\end{lemma}
\begin{proof}
Fix $n\in\N$. Let $\varepsilon>0$ be such that $i\varphi \notin [1-\varepsilon,1)$ for $i=-1,0,1,2,\dots, n-1$. Let $k$ be the smallest natural number such that $k\varphi\in [1-\varepsilon,1)$ and $X$ be a Sturmian word generated by $\varphi$ with initial point $\varepsilon$. We define $\widetilde{w}(i)=\widetilde{w}(i+k)=X(i)$ (and $w(i)=X(i)$). We need to show that $X(i+k)=X(i)$ for $0\leq i\leq k-1$. 

By definition, if $X(i)=0$ then $i\varphi +\varepsilon \in [0,\varphi)$ and by assumption $i\varphi\notin [1-\varepsilon,1)$ so $i\varphi +\varepsilon\in [\varepsilon,\varphi)$. Hence $i\varphi+ k\varphi +\varepsilon \in [0, \varphi)$ which means that $X(i+k)=0$. If $X(i)=1$ then $i\varphi +\varepsilon \in [\varphi,1)$, and since $(i-1)\varphi\notin [1-\varepsilon,1]$ we get that $i\varphi +\varepsilon \in [\varphi+\varepsilon,1)$. This gives $i\varphi+k\varphi+\varepsilon\in [\varphi, 1)$ which means that $X(i+k)=1$. 
\end{proof}

\begin{theorem}\label{sturmian_unstable}
Assume that $\varphi \in (0,1)$ is badly approximable and the interaction energy decays as $1/r^{\alpha}$ with $\alpha>3$. 
Then Sturmian ground-state configurations generated by $\varphi$ are unstable.  
\end{theorem}
\begin{proof}
Let $X$ be a Sturmian sequence generated by $\varphi$. Let us now introduce a small chemical potential $\mu$ which favors the presence of particles (occurrence of 1's). For every $m\in \N$ pick $k_m\geq m$ and a word $w_m$ of length $k_m$ whose doubling is a Sturmian subword as in Lemma \ref{lemma-doubling}. Let $Y_m$ be the periodic word with period $w_m$. For each $i$ such that $Y_m(i)=1$ the energy from pairs of 1's containing $Y_m(i)$ is not greater than $$2\sum_{i=k_m+1}^\infty \frac{1}{i^\alpha}\leq 
2\int_{k_m}^\infty \frac{dx}{x^\alpha} = \frac{2}{\alpha-1} k_m^{-\alpha+1}$$
since there are no forbidden pairs of 1's in $Y_m$ at distance less than $k_m$.
Hence the difference of energies coming from pairs of 1's between $Y_m$ and $X$ involving particles in the interval $[-l,l]$ may be estimated as follows. 
\begin{gather}\label{H2-energy}
E_1(Y_m([-l,l]))-E_1(X([-l,l]))= E_1(Y_m([-l,l]))\leq (2l+1) \frac{2}{\alpha-1} k_m^{-\alpha+1}.
\end{gather}
Denote by $\xi(X)$ and $\xi(Y_m)$ the frequency of 1's in $X$ and $Y_m$ respectively. We have $\xi(X)=1-\varphi$ and $\xi(Y_m)= \frac{n(w_m)}{k_m}$ (where $n(w_m)$ is the number of 1's in $w_m$). Since $\varphi$ is badly approximable there exists $c>0$ independent of $m$ such that 
$$|\xi(Y_m)- \xi(X)| = \left |\frac{n(w_m)-k_m}{k_m} + \varphi \right | > \frac{c}{k_m^2}.$$
This shows that 
\begin{gather*}
    |E_2(Y_m([-l,l])) - E_2(X([-l,l]))|\geq (2l+1)\frac{|\mu|c}{2k_m^2},
\end{gather*}
 where $E_2$ is the energy of 1's in the interval $[-l,l]$ (from the chemical potential $\mu$).

Replacing $\mu$ with $-\mu$ if necessary, we get
\begin{gather}\label{H2lambda-energy}
 E_2(Y_m([-l,l])) - E_2(X([-l,l]))\leq -(2l+1)  \frac{|\mu|c}{2k_m^2}.
\end{gather}
Combining (\ref{H2-energy}) and (\ref{H2lambda-energy}) gives 
\begin{gather*}
    E(Y_m([-l,l]))-E(X([-l,l]))\leq (2l+1) \frac{2}{\alpha-1} k_m^{-\alpha+1}- (2l+1)  \frac{|\mu|c}{2}k_m^{-2}.
\end{gather*}
By dividing by $2l+1$ and taking the limit $l\rightarrow \infty$ we obtain
\begin{gather*}
    \rho(Y_m)-\rho(X) \leq \frac{2}{\alpha-1} k_m^{-\alpha+1}-   \frac{|\mu|c}{2}k_m^{-2},
\end{gather*}
where $\rho(Y_m), \rho(X)$ denote the energy densities of $Y_m$ and $X$ respectively.

Since $\alpha>3$, $\frac{2}{\alpha-1}k_m^{-\alpha+1}$ tends to $0$ faster than $\frac{|\mu|c}{2}k_m^{-2}$ when $m\rightarrow \infty$ so for large enough $m$ we have
$$\rho(Y_m)<\rho(X),$$
which completes the proof.
\end{proof}

\section{Discussion}

We studied stability of one-dimensional non-periodic ground-state configurations in classical lattice-gas models with interactions decaying as $1/r^{\alpha}$
with respect to finite-range perturbations of interactions. We showed that the Thue-Morse ground state is unstable for any $\alpha > 1$ 
and the Sturmian ground states are unstable for $\alpha > 3$.

It is a fundamental problem to construct a one-dimensional lattice-gas model with the unique non-periodic ground state
which is stable with respect to finite-range perturbations of interactions. We conjecture that Sturmian ground states generated by rotations on the circle 
by badly approximable irrationals are stable for some small values of $\alpha$. 
\vspace{3mm}

{\bf Acknowledgments} We would like to thank the National Science Centre (Poland) for financial support under Grant No. 2016/22/M/ST1/00536.

\end{document}